\documentclass[12pt]{article}

\usepackage{natbib}
\setcitestyle{authoryear,round,notesep={; },aysep={},yysep={;}}

\usepackage{amsmath,amssymb,amsfonts,amsthm,mathrsfs}
\usepackage[dvipsnames]{xcolor}
\usepackage{hyperref}
\usepackage{tabularx}
\usepackage{graphicx}
\usepackage{multirow}
\usepackage{booktabs}
\usepackage{textcomp}
\usepackage{epstopdf}
\usepackage{lipsum}
\usepackage{algorithm}
\usepackage{algorithmicx}
\usepackage{algpseudocode}
\usepackage{listings}
\usepackage{subcaption}
\usepackage{makecell}
\usepackage{gensymb}
\usepackage{tabularx}
\usepackage{epsfig, graphics, graphicx}
\usepackage{epstopdf}
\usepackage{tikz}
\usepackage{multirow}
\usepackage{color}

\makeatletter
\newcommand{\removeperiod}{\@ifnextchar.{\@gobble}\relax}
\makeatother
\hypersetup {
    colorlinks = true,
    linkcolor = blue,
    filecolor = magenta,
    urlcolor = NavyBlue,
    citecolor = NavyBlue
}

\usepackage{geometry}
\newtheorem{theorem}{Theorem}[section]
\newtheorem{lemma}[theorem]{Lemma}
\theoremstyle{remark}
\newtheorem{remark}[theorem]{Remark}

\makeatletter
\def\th@remark{%
  \thm@headfont{\normalfont\bfseries}
  \thm@notefont{\normalfont\bfseries}
  \thm@bodyfont{\normalfont}
  \thm@headpunct{.}
}
\makeatother
\newcommand{\eps}{\varepsilon}
\usepackage{authblk}

\title{Derivations of Animal Movement Models with Explicit Memory}
\author{%
    Tianxu Wang$^{*}$, Kyunghan Choi$^{*}$, Hao Wang$^{\dagger}$%
}
\affil[]{Interdisciplinary Lab for Mathematical Ecology \& Epidemiology, University of Alberta, Edmonton, Alberta T6G 2G1, Canada}
\affil[]{Department of Mathematical and Statistical Sciences, University of Alberta, Edmonton, Alberta T6G 2G1, Canada}

\begin{document}

\maketitle

\renewcommand{\thefootnote}{\fnsymbol{footnote}} 
\footnotetext[1]{Co-first authors}
\footnotetext[2]{Corresponding author: hao8@ualberta.ca}

\begin{abstract}
1. Highly evolved animals continuously update their knowledge of social factors, refining movement decisions based on both historical and real-time observations. Despite its significance, research on the underlying mechanisms remains limited. 


2. In this study, we explore how the use of explicit memory shapes different mathematical models across various ecological dispersal scenarios. Specifically, we investigate three memory-based dispersal scenarios: gradient-based movement, where individuals respond to environmental gradients; environment matching, which promotes uniform distribution within a population; and location-based movement, where decisions rely solely on local suitability. These scenarios correspond to diffusion advection, Fickian diffusion, and Fokker-Planck diffusion models, respectively.


3. We focus on the derivation of these memory-based movement models using three approaches: spatial and temporal discretization, patch models in continuous time, and discrete-velocity jump process. 
These derivations highlight how different ways of using memory lead to distinct mathematical models. 

4. Numerical simulations reveal that the three dispersal scenarios exhibit distinct behaviors under memory-induced repulsive and attractive conditions. The diffusion advection and Fokker-Planck models display wiggle patterns and aggregation phenomena, while simulations of the Fickian diffusion model consistently stabilize to uniform constant states.
\end{abstract}

\makeatother

\section{Introduction}
Animal movement is not purely random; it is often influenced by internal factors, particularly memory. Many animals rely on memory to guide their movements, storing and utilizing information about their environment over time. This spatiotemporal memory helps them navigate more efficiently, as they continuously update their understanding of their surroundings. For example, pigeons can recall and report past actions \citep{zentall2001episodic}, and chimpanzees can use lexigrams to indicate food locations hours after observing baiting events \citep{terrace2005missing}. These behaviors illustrate how memory-based movement decisions enable animals to refine their navigation strategies based on past experiences \citep{schlagel2014detecting, martin2010keeping}.
Moreover, memory-driven movement is often influenced by social factors. Although animals rely on their individual memory, the behaviors of the group can also guide their movements. For example, animals may adjust their own movement strategies based on the observed behavior of others, particularly in species where group cohesion or social interactions play a significant role \citep{ford2013two,matthysen2005density,cressman2011effects,allee1927animal}. This interaction between individual and collective memory helps animals navigate within social groups, ensuring more efficient group dynamics.

While earlier studies largely focused on the relationships between static environmental factors and movement behavior, it has become increasingly clear that many animals possess the ability to acquire, store, and utilize information about their surroundings over time \citep{stevens1997aggregation}. This information is not only based on fixed snapshots of the environment but also on temporal changes that accumulate over time. Highly evolved animals continuously update their understanding of environmental dynamics, refining their movement decisions based on past and current observations. This ability to integrate temporal information allows animals to adapt their strategies to changing conditions, making their navigation more efficient \citep{owen2010foraging, bar2011use}. Despite its importance, research exploring the role of cumulative temporal information in movement remains limited. Furthermore, the dispersal strategies of animals based on explicit memory vary depending on the situation. It is important to consider different types of movement alongside memory when developing models.

In recent years, movement models have expanded to include different ways that animals process and respond to their environment \citep{fagan2013spatial, avgar2013empirically}. These models provide a tool for understanding how animals navigate based on cognitive and contextual factors.
A mathematical form of the context-dependent movement model with a certain environmental cue $v=v(x,t)$ can be formulated as the generalized equation
\begin{equation*}\label{eq: generalized model}
    u_t = \nabla\cdot (D(v)\nabla u + A(v)u\nabla v),
\end{equation*} where $u$ is the population density of an animal, $D(s)>0$ and $A(s)$ can be positive or negative for all $s\ge 0$. $v$ represents a quantity corresponding to the context in the model. There have been many studies with specific choices of $D$, $A$ and $v$ \citep{Winkler2005,Winkler2010,Yoon2017FP,YJK2019} depending on situations. The determination of 
$D$ and $A$ are influenced by dispersal scenarios of animals: first, measuring the gradient of environmental factors, $\nabla v$, at the current location to guide movement (Diffusion-Advection equations); second, matching environmental cues of the current location with those of surrounding areas to decide movement (Fickian diffusion); and third, making movement decisions based solely on the environmental factors, the quantity of $v$, at the current location (Fokker-Planck diffusion). These methods provide various ways to model movement behaviors in ecological and cognitive contexts.

\begin{figure}
    \centering
    \includegraphics[width=0.8\linewidth]{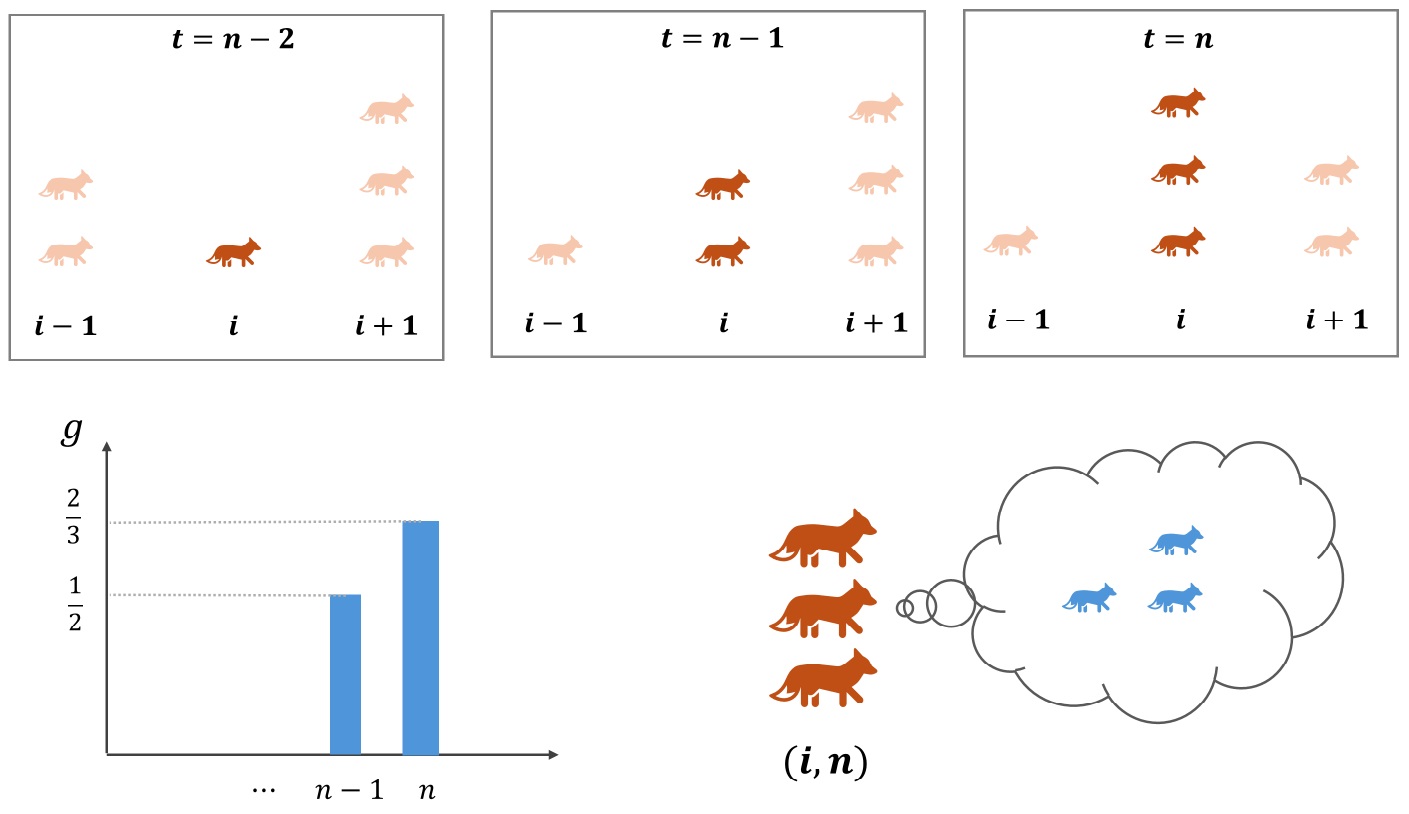}
    \caption{An illustration of the memory acquisition process in foxes is shown. The three figures at the top depict the movement of six foxes across spatial locations \(i-1\), \(i\), and \(i+1\) over time steps \(t = n-2\), \(n-1\), and \(n\). The bottom-left graph represents memory capacity at each time step, showing that the foxes retain memories only from the recent times \(t = n-1\) and \(n\).  
If \(u_i^n\) denotes the population at \((i, n)\), the three foxes at location \(i\) and time \(n\) have a memory quantified as 
\(
v^n_i = g(n-1) u_i^{n-1} + g(n) u_i^n =\frac{1}{2} u_i^{n-1} + \frac{2}{3} u_i^n,
\)
which equals 3 in this figure.}
    \label{fig:memory process}
\end{figure}
To address the continuous acquisition of memory from social behavior, we think of \( v \) as information accumulated continuously over time, and formulated in the form of 
\begin{equation*}
\label{eq: g convolution m}
    v(x,t)=(g*_tu)(x,t) = \int_0^tg(t-s)u(x,s)ds,
\end{equation*} where $g$ is the temporal kernel.  The information from different times does not carry the same significance. The discretized version of the memory-gaining process can be seen in Figure \ref{fig:memory process}. In addition, animals have different memorizing mechanisms, so the choice of kernel can vary. Typically, more recent information is more important, and this varying importance is captured by a weighting function \( g(t) \). For example, if take an exponentially decaying kernel $g(t) = e^{-\mu t}$ with $\mu>0$, then 
$v$ represents the exponentially weighted historical population density.
One can easily show that $u$ and $v$ satisfy the following coupled PDE-ODE system:
\begin{equation*}
    \begin{cases}
        u_t = \nabla \cdot \big( D(v) \nabla u + A(v) u \nabla v \big), \\
        v_t = u - \mu v,
    \end{cases}
\end{equation*}
where $D$ and $A$ are functions of $v$. Similar systems, with specific choices of $D$ and $A$, have been analyzed in the context of chemotaxis models, particularly with a diffusion of $v$~\citep{Winkler2005, Yoon2017FP, Winkler2010}. From the perspective of memory models, $v(x,t)$ can be interpreted as a memory map acquired by the population $u(x,t)$ at position $x$ and time $t$ by marking their territories, which decays over time at a rate $\mu$ \citep{Potts2019,Wang2023}. The model then explains the movement of the animals by using local marking information $v(x,t)$ from the memory map.  
In the case of fixed memory at a certain time, where   $g=\delta_\tau$ (a delta distribution at $\tau$), 
Shi et al. studied a diffusion advection model ($D(s)\equiv D$ and  
$A(s)\equiv \alpha$) through delayed density distributions \citep{shi2017hopf}:
\[
u_t = D \Delta u + \alpha \nabla\cdot (u\nabla (\delta_\tau*_tu)), \quad x \in \Omega, \, t > 0,
\]
where $\tau > 0$ is a fixed time delay and $\delta_\tau*_tu = u(x,t-\tau)$. Here, the cognitive movement driven by episode-like memory is represented by \(\alpha\nabla\cdot (u\nabla (\delta_\tau*_tu))\), which assumes that movement decisions are based on information available at a specific time, $\tau$ time ago. This idea has since inspired further work on the dynamics of fixed delayed-time information \citep{wang2023dynamics, song2022spatiotemporal, song2019spatiotemporal, shi2020diffusive, shi2019diffusive, zhang2023diffusive}.

The purpose of the paper is to derive memory-dependent diffusion models that depend on the temporally distributed population, rewritten by a general equation
\begin{equation}\label{eq: main equation}
    u_t = \nabla\cdot (D(g*_tu)\nabla u + A(g*_tu)u\nabla(g*_tu)),
\end{equation} where $D(s)>0$ and $A(s)$ can be positive or negative for all $s\ge 0$, and 
\[
g*_tu = \int_0^tg(t-s)u(x,s)\,ds.
\] $u(x,t)$ is the density of an animal that can gain and memorize its temporally distributed density $g*_tu$ with a positive kernel function $g$. The model describes the movement of animals sharing and using accumulated information about population density with a certain type of dispersal strategy. Specifically, three different dispersal ways determine different choices of $D$ and $A$:
\begin{itemize}
    \item [(i)] $D(s)\equiv D$ and $A(s)\equiv \gamma'(s)$, a diffusion advection equation: \[ u_t = D\Delta u +  \nabla\cdot(u\nabla\gamma(g*_tu));\]
    \item [(ii)] $D(s)\equiv \gamma(s)$ and $A(s)\equiv 0$, the Fick's law diffusion equation: \[u_t = \nabla\cdot(\gamma(g*_tu)\nabla u); \]
    \item [(iii)] $D(s)=\gamma(s)$ and $A(s)= \gamma'(s)$, the Fokker-Planck type diffusion equation: \[ u_t = \Delta (\gamma(g*_tu)u) . \]
\end{itemize} These diffusion models illustrate three distinct movement strategies, all based on the same memorized quantity \( g *_t u \). The function \(\gamma\) represents a dispersal rate influenced by memory. When \(\gamma\) is an increasing function, animals tend to move toward areas with lower memory. Conversely, if \(\gamma\) is a decreasing function, they tend to move toward areas with higher memory.  
In the first case, animals measure the gradient of temporal memory, \(\nabla \gamma(g *_t u)\), and base their directed movement decisions on this gradient. The second case describes scenarios where the dispersal rate is matched with the surroundings, depending on their memory. The third case represents movement determined solely by the value of \(\gamma(g *_t u)\), without reference to its gradient.
We derive these three equations in Section~2, demonstrating different approaches to how animals decide their movement based on memory.


The structure of the paper is as follows. In Section \ref{sec:dispersal scenarios}, we discuss the ecological implications of the three diffusion types. In Section \ref{sec:Formal derivations}, we formally derive these three memory-dependent diffusion models using three different approaches. First, we discretize both time and space variables and study the recursive relation between past and present states by considering the discrete accumulated memory over time. Next, we distinguish the three equations using patch-based models to explore the implications of different diffusion types on movement. Finally, we derive the models through the discrete velocity jump process. In Section \ref{sec:Numerical simulation}, we compare the performance of the three diffusion types through numerical simulations. Additionally, the well-posedness of the generalized model is provided in Appendix \ref{app:wellposedness}, and the linear stability for a special case is addressed in Appendix \ref{app:linear stability}.

\section{Three memory-based dispersal scenarios}
\label{sec:dispersal scenarios}

In this section, we discuss how three ecological movement models describe animal movement based on explicit memory. Organisms employ various dispersal strategies in ecological systems when interacting with their environment, and these strategies can be captured through different modeling approaches. Diffusion models serve as a foundational framework for understanding movement behaviors. Constant diffusion, for instance, has explained random animal movement. However, in heterogeneous environments, diffusion processes become more complex, requiring models that account for spatial variability \citep{Hoyoun2024HD}. An animal's decision is influenced by cumulative memory across different locations. Non-homogeneous memory introduces biased movement, amplifying spatial heterogeneity in population distribution. Therefore, understanding the relationship between memory and ecological dispersal strategies is essential for accurately modeling such behaviors. We introduce three ecological dispersal scenarios.

\begin{itemize}
    \item [(i)] Gradient-Based Movement: organisms measure the gradient of environmental factors at their current location to guide movement. This approach captures the active response of individuals to changes according to their surroundings, incorporating both random roaming and directed movement along population gradients. For example, migratory birds, such as geese, use memory to recall landmarks and environmental cues, such as wind patterns or temperature gradients, during long-distance migrations \citep{zein2022simulating}. This ability to remember key factors helps them navigate effectively. Similarly, most territorial animals can sense population pressure and decide whether to move or stay. When roaming within a small area, they compare historical environmental factors in the vicinity with those at their current position to determine whether to remain or move to a different location. 
    
    Mathematically, this behavior is modeled using diffusion advection equations, where memory-driven movement decisions depend on the spatial gradient of accumulated environmental information. Such models are particularly useful for studying scenarios where animals respond to spatially varying conditions.
    
    
\item [(ii)] Environment Matching: In this scenario, organisms exhibit group-level behavior by matching their dispersal rate with individuals in their vicinity to achieve a uniform population distribution. They compare environmental cues within the group and coordinate their movement to achieve an even distribution that optimizes resource use. Animals use this strategy to avoid overcrowding and maintain steady access to food. Their collective memory of past group formations helps them balance the environment’s resources and maintain stability.

This scenario is modeled by Fickian diffusion equations, where the dispersal rate matches the local environmental conditions. Unlike gradient-based movement, which involves directional bias, environment matching emphasizes homogeneity, reducing spatial heterogeneity in population distribution.
    
   \item [(iii)] Location-Based Movement: In this case, animals decide whether to move or stay based on the suitability of their current location. Unlike gradient-based movement, this strategy does not involve comparing the current location with other areas; instead, it focuses on the animal's satisfaction with its immediate surroundings. Memory plays a crucial role in helping animals evaluate whether similar past conditions led to positive outcomes. For instance, territorial animals assess their location based on factors such as prey availability, safety, and environmental conditions \citep{clobert2009informed}. If the location satisfies their needs, they stay; otherwise, they move to explore better territories. When population density is a key factor in determining movement, animals rely on past cumulative experiences to evaluate the population abundance at a given location. If the current abundance meets their needs, they stay; otherwise, they move to explore more favorable conditions.

   This behavior is captured by Fokker-Planck diffusion equations, which focus on the influence of local environmental memory. Unlike gradient-based movement, this strategy depends solely on the local value of memory, without reference to its spatial gradient.

\end{itemize}

When modeling ecological dispersal with memory, it is crucial to accurately capture how memory interacts with social behaviors to influence movement. These scenarios provide diverse perspectives for understanding memory-based dispersal in ecological and cognitive contexts. In the next section, we derive three memory-based diffusion models corresponding to the three ecological dispersal strategies using different mathematical methods.

\section{Formal derivations of explicit memory models}
\label{sec:Formal derivations}

Most of the studies on cognitive movement are mainly based on diffusion advection frameworks, represented as \(\Delta u + \nabla(u\nabla \gamma(x))\). However, alternative forms of diffusion could produce substantially different outcomes in heterogeneous environments.  For example, Fick’s law, $\nabla (\gamma(x)\nabla u)$, is commonly applied to physical diffusion influenced by geometric or boundary conditions \citep{fick1855ueber, bringuier2011particle, okubo2001diffusion}. In contrast, the Fokker-Planck diffusion operator $\Delta (\gamma(x)u)$ is widely used in ecological models to describe processes such as random walks or population dynamics guided by local environmental cues \citep{fokker1914mittlere, planck1917satz, okubo2001diffusion, cantrell2006movement, Yoon2017FP, JieLaurencot2021chemo, WangKim2022}.
Depending on the biological context, one can select the most appropriate operator to accurately model movement behaviors aligned with specific environmental or ecological scenarios.

In this section, we derive the memory-induced diffusion models using three distinct methods:  discretization of space and time, patch-model approach, and discrete velocity jump processes.


\subsection{One-Dimensional Discrete Method}

We now derive the equation through the discretization method. Let the one-dimensional spatial domain \( [0, L] \) be divided into discrete points indexed by \( l \), and the temporal domain \( [0, t] \) (up to the current time) be divided into \( N \) discrete steps indexed by \( k \). Denote the population density at position \( l \) and time \( k \) as \( u_l^k \).

\begin{center}
\begin{tikzpicture}[>=stealth, scale=1.5]

\draw[gray, thick] (-1.5, 1.6) -- (-1.5, -1.6);
\draw[gray, thick] (-0.5, 1.6) -- (-0.5, -1.6);
\draw[gray, thick] (0.5, 1.6) -- (0.5, -1.6);
\draw[gray, thick] (-3, 1) -- (2, 1);
\draw[gray, thick] (-3, 0) -- (2, 0);
\draw[gray, thick] (-3, -1) -- (2, -1);

\draw[->, thick] (-0.5, 1.8) -- (-0.5, -1.9);
\draw[->, thick] (-3, -1) -- (2, -1);

\node[draw, circle, fill=white!20, minimum size=3pt, align=center] at (-1.5, 1) {\(u_{l-1}^{N-2}\)};
\node[draw, circle, fill=white!20, minimum size=3pt, align=center] at (-0.5, 1) {\(u_l^{N-2}\)};
\node[draw, circle, fill=white!20, minimum size=3pt, align=center] at (0.5, 1) {\(u_{l+1}^{N-2}\)};
\node[draw, circle, fill=white!20, minimum size=3pt, align=center] at (-1.5, 0) {\(u_{l-1}^{N-1}\)};
\node[draw, circle, fill=white!20, minimum size=3pt, align=center] at (-0.5, 0) {\(u_l^{N-1}\)};
\node[draw, circle, fill=white!20, minimum size=3pt, align=center] at (0.5, 0) {\(u_{l+1}^{N-1}\)};
\node[draw, circle, fill=white!20, minimum size=3pt, align=center] at (-1.5, -1) {\(u_{l-1}^{N}\)};
\node[draw, circle, fill=blue!20, minimum size=3pt, align=center] at (-0.5, -1) {\(u_l^N\)};
\node[draw, circle, fill=white!20, minimum size=3pt, align=center] at (0.5, -1) {\(u_{l+1}^{N}\)};

\node[below] at (-1.5, 2.3) {\(l-1\)};
\node[below] at (-0.5, 2.3) {\(l\)};
\node[below] at (0.5, 2.3) {\(l+1\)};
\node[below] at (-0.5, -1.9) {\(t\)};
\node[below] at (2.2, -0.83) {\(x\)};

\end{tikzpicture}
\end{center}

For an individual in motion, each step is influenced by information accumulated from all past time points. This accumulated information at position $l$ can be modeled by
\[
\tilde u_l=\frac{t}{N} \sum_{k=1}^N \phi^k u_l^k,
\]
where \( \phi^k \) is a weight function representing the importance of the information at time point \( k \), and \( \frac{t}{N} \) is the time step length. In this section, we gave an example using the aggregation case, assuming animals have a tendency to move toward higher population
densities. For a repulsion case, one can use a similar argument.

\subsubsection{Diffusion-Advection model}
Consider that the relative population density influences the probability of movement in the vicinity. For instance, if the population aggregation at position \( l+1 \) is greater than that at \( l \), individuals at \( l \) are more likely to move toward position \( l+1 \). This movement tendency depends not only on the current population density information but also on historical information. The difference, accumulated over all past time, between the weighted population densities at \( l+1 \) and \( l \), is given by
\begin{equation*}
\label{eq:tilde_u}
    \tilde{\Tilde{u}}_{l} = \tilde u_{l+1} - \tilde u_l = \frac{t}{N} \left(\sum_{k=1}^N \phi^k u_{l+1}^k - \sum_{k=1}^N \phi^k u_l^k \right).
\end{equation*}
The probabilities for an individual at position \( l \) to move left or right are then
\begin{align*}
    P_{\text{left}} &= p - (1-p) \tilde{\tilde{u}}_{l-1}, \\
    \label{eq:P_right}
    P_{\text{right}} &= p + (1-p) \tilde{\tilde{u}}_l,
\end{align*}
where \( 0 < p < 1 \) represents random diffusion driven movement, and \( (1-p) \tilde{\Tilde{u}}_i \)  accounts for directed cognitive movement influenced by historical population information. Hence, if the accumulated historical population at \( l+1 \) is greater than that at \( l \), i.e., \( \tilde{\tilde{u}}_l > 0 \), then \( P_{\text{right}} > p > P_{\text{left}} > 0 \), indicating an increasing trend toward the right. The change of population $\delta u_l$ at position \( l \) is determined by individuals leaving \( l\) and individuals arriving from \( l-1 \) and \( l+1 \):
\begin{equation}
\label{eq:delta_u}
\begin{aligned}
    \delta u_l &= - \left(p - (1-p)\tilde{\Tilde{u}}_{l-1}\right) u_l - \left(p + (1-p)\tilde{\Tilde{u}}_l\right) u_l 
    \\
    & \quad + \left(p + (1-p)\tilde{\Tilde{u}}_{l-1}\right) u_{l-1} + \left(p - (1-p)\tilde{\Tilde{u}}_l\right) u_{l+1}.
    \end{aligned}
\end{equation}
Linearize \( u_{l\pm 1} \) and \( u_{l\pm 1}^k \) around \( u_l \) and \( u_l^k \) respectively, we obtain 
\begin{equation}
\label{eq: Tylor expansion}
\begin{aligned}
    u_{l\pm 1} &= u_l \pm h \left(u_l\right)_x + \frac{h^2}{2} \left(u_l\right)_{xx} + O(h^3), \\
    u_{l\pm 1}^k &= u_l^k \pm h \left(u_l\right)_x^k + \frac{h^2}{2} \left(u_l\right)_{xx}^k + O(h^3).
\end{aligned}
\end{equation}
For notation simplicity, we use $u$ for $u_l$ in the following. Substitute approximations \eqref{eq: Tylor expansion} into \eqref{eq:delta_u}, giving
\begin{equation*}
    \delta u = ph^2 u_{xx} - 2(1-p) h^2 \frac{t}{N} \sum_{k=1}^N \phi^k \left(u_x^k u\right)_x.
\end{equation*}
Consider population changes \(  \delta u \) in a short time interval \( \tau \), then \(u_t = \lim\limits_{\tau \to 0} \frac{\delta u}{\tau}\). Let \( D = p \frac{h^2}{\tau} \) and \( \alpha = 2(1-p) \frac{h^2}{\tau} \). Dividing each side by \( \tau \), and letting $\tau \to 0$, we obtain
\begin{equation*}
    u_t = D u_{xx} - \alpha \frac{t}{N} \sum_{k=1}^N \phi^k \left(u_x^k u\right)_x,
\end{equation*}
which is a one-dimensional continuous diffusion advection model. The corresponding higher-order version is
\begin{equation}
\label{eq: discrete DA}
    u_t = D \Delta u  - \alpha \nabla \cdot \left( u \nabla \left(g*_t u\right)\right).
\end{equation}
For a more general case, using a similar argument and substituting \( \tilde{u}_l \) with \( \gamma(\tilde{u}_l) \), one can derive  
\[
u_t = D\Delta u + \nabla \cdot \big(u \nabla \gamma(g *_t u)\big),
\]
where the utilization of information is more general and is modeled using the function \( \gamma \). Equation \eqref{eq: discrete DA} is a special case when $\gamma$ is linear.

\subsubsection{Fickian type diffusion}
If dispersal is symmetric, i.e., the probability of an individual moving from \( l \) to \( l+1 \) is the same as the probability of moving from \( l+1 \) to \( l \), then the probabilities for an individual at position \( l \) to move left or right are given by
\begin{align*}
    P_{\text{left}} &= p - (1-p) \tilde u_{l-1}, \\
    P_{\text{right}} &= p - (1-p) \tilde u_l. 
\end{align*} 
This implies that individuals do not move purely randomly; instead, they adjust their direction based on the accumulated information at their current position and nearby locations. By a similar argument as above, 
applying expansions \eqref{eq: Tylor expansion}, we obtain
\begin{equation*}
    \delta u =  p h^2 u_{xx} - (1-p)h^2 \frac{t}{N} \sum_{k=1}^N \phi^k \left(u^k u_x\right)_{x}.
\end{equation*}
Let \( D = \frac{p h^2}{\tau} \) and \( \alpha = \frac{(1-p) h^2}{\tau} \). Dividing both sides by \( \tau \) and taking the limit as \( \tau \to 0 \), the population change rate is then given by
\begin{equation*}
    u_t = D u_{xx} - \alpha \frac{t}{N} \sum_{k=1}^N \phi^k \left(u^k u_x\right)_{x}.
\end{equation*}
This is a Fickian type diffusion equation.
The corresponding higher dimensional continuous version is
\begin{equation}
\label{eq: discrete FI}
    u_t = D \Delta u - \alpha \nabla \cdot \left(\left(g *_t u\right) \nabla u \right) =  \nabla \cdot \left( \left(D - \alpha \left(g *_t u\right) \right)\nabla u \right).
\end{equation}
Here $D$ is the random diffusion coefficient and $\alpha$ is the directed sensitivity coefficient. With $\alpha$ increasing, animals have stronger directed cognitive movement based on temporal population distribution information.
By substituting \( - \tilde{u}_l \) with \( \gamma(\tilde{u}_l) \) and setting \( p = 0 \), one can obtain 
\[
u_t = \nabla \cdot \big( \gamma(g *_t u) \nabla u \big).
\]
This is a more general form of Fickian type diffusion, animals can obtain and use historical information as function $\gamma(g *_t u) $. Equation \eqref{eq: discrete FI} is a special case when $\gamma(z)= D-\alpha z$. 

\subsubsection{Fokker-Planck type diffusion}
If movement is random, i.e. at any position, the probabilities of an individual moving to the left or right are equal, then for an individual at position \( l \), we have
\begin{align*}
    P_{\text{left}} = P_{\text{right}} = p - (1-p) \Tilde{u}_{l}.
\end{align*}
This probability reflects another movement strategy of individuals, where they adjust their movement based on the accumulated local population information \( \Tilde{u}_{l} \).
Applying expansions \eqref{eq: Tylor expansion}, we obtain
\begin{equation*}
    u_t = D u_{xx} - \alpha \frac{t}{N} \sum_{k=1}^N \phi^k \left(u^k u\right)_{xx},
\end{equation*}
where \( D = \frac{p h^2}{\tau} \) and \( \alpha =  \frac{(1-p) h^2}{\tau} \). This is a one-dimensional Fokker Planck type diffusion equation. The corresponding higher-dimensional version is
\begin{equation*}
    u_t = D \Delta u - \alpha \Delta 
 (u \left(g *_t u\right)) = \Delta\left(\left( D - \alpha  \left(g *_t u\right)\right) u \right).
\end{equation*}
Similarly, by substituting \( -\tilde{u}_l \) with \( \gamma(\tilde{u}_l) \) and setting \( p = 0 \), we have  
\[
u_t = \Delta \big( \gamma(g *_t u)  u \big).
\]
The repellent case, in which animals tend to move away from regions of higher population density, can be derived using the same argument with a change in the sign of $\alpha$.

\subsection{Patch models}
In this section, we consider a continuous time model over discrete patches.
\begin{center}
\begin{tikzpicture}[thick, >=stealth]

\node[draw, circle, fill=blue!30, minimum size=1.5cm] (uim1) at (0, 0) {\(u_{i-1}\)};
\node[draw, circle, fill=blue!30, minimum size=1.5cm] (ui) at (3, 0) {\(u_i\)};
\node[draw, circle, fill=blue!30, minimum size=1.5cm] (uip1) at (6, 0) {\(u_{i+1}\)};

\draw[->, line width=1.5pt] (uim1) .. controls (1.5, 0.5) .. (ui) 
    node[midway, above] {\(c_{i i-1}\)};
\draw[<-, line width=1.5pt] (uim1) .. controls (1.5, -0.5) .. (ui) 
    node[midway, below] {\(c_{i-1i}\)};

\draw[->, line width=1.5pt] (ui) .. controls (4.5, 0.5) .. (uip1) 
    node[midway, above] {\(c_{i+1 i}\)};
\draw[<-, line width=1.5pt] (ui) .. controls (4.5, -0.5) .. (uip1) 
    node[midway, below] {\(c_{i i+1}\)};

\end{tikzpicture}
\end{center}
Here, \( u_i \) represents the population at patch \( i \), and \( c_{ij} \, (c_{i \leftarrow j}) \) denotes the migration rate or departure probability from patch \( j \) to patch \( i \). The changes in the population at patch \( i \), are governed by the following equation:
\begin{equation*}
    \dot{u}_i = c_{i i-1} u_{i-1} + c_{i i+1} u_{i+1} - c_{i-1 i} u_i - c_{i+1 i} u_i.
\end{equation*}
Let $v_i = g *_t u_i$ represent the accumulated information of the population at patch \( i \) over all past time, where \( g \) is a weighting function that assigns importance to information based on its temporal proximity.  

\subsubsection{Diffusion-Advection model}

The movement of individuals is influenced by temporal information both in their current patch and in their destination patch. We consider the following two cases for different forms of temporary information utilization.

\textbf{Case 1.}  
Let \( \gamma(v_j - v_i) \) represent the influence of accumulated temporal information from the vicinity. The departing probabilities are then defined as  
\begin{equation*}
    c_{ij} = p_1 + \gamma(v_j - v_i),
\end{equation*}
where \( \gamma \) is an odd function, and \( p_1 \) represents the diffusion-driven component. If \( \gamma \) is an increasing function, it indicates that animals tend to move away from regions of higher population density. Conversely, if \( \gamma \) is a decreasing function, it signifies aggregation behavior, where individuals move toward areas of higher population density.
The population dynamics at patch \( i \) are then governed by
\begin{equation*}
\begin{aligned}
    \dot{u}_i &= p_1(u_{i-1} + u_{i+1} - 2 u_i) + \gamma(v_{i+1} - v_i) u_{i+1} - \gamma(v_i - v_{i-1}) u_{i-1} \\
    &\quad + (\gamma(v_{i+1} - v_i) - \gamma(v_i - v_{i-1})) u_i.
\end{aligned}
\end{equation*} This can be rewritten as
\begin{equation*}
\begin{aligned}
    \dot{u}_i &= p_1(u_{i-1} + u_{i+1} - 2 u_i) + \gamma(v_{i+1} - v_i)(u_{i+1} + u_i) - \gamma(v_i - v_{i-1})(u_{i-1} + u_i).
\end{aligned}
\end{equation*}
The corresponding continuous differential equation is
\begin{equation*}
    u_t = p_1 \Delta u + 2 \gamma'(0) \nabla \cdot (u \nabla v).
\end{equation*}
Note that only $\gamma'(0)$ is needed, and we don't have more requirements for $\gamma$.
Let \( D = p_1 \) and \( \alpha = 2 \gamma'(0) \). Substituting \( v = g *_t u \), we obtain
\begin{equation*}
    u_t = D \Delta u + \alpha \nabla \cdot (u \nabla \left(g*_t u\right)).
\end{equation*}
\textbf{Case 2.}  
Alternatively, if the temporal information is represented by \( \gamma(v_i) - \gamma(v_{i-1}) \), then the departure probabilities are given by
\begin{equation*}
    c_{ij} = p_1 + \gamma(v_j) - \gamma(v_i).
\end{equation*}
It follows that
\begin{equation*}
\begin{aligned}
    \dot{u}_i &= p_1(u_{i-1} + u_{i+1} - 2 u_i) + (\gamma(v_{i+1}) - \gamma(v_i))u_{i+1} - (\gamma(v_{i}) - \gamma(v_{i-1}))u_{i-1} \\
    & + (\gamma(v_{i+1}) + \gamma(v_{i-1}) - 2\gamma(v_i)) u_i.
    \end{aligned}
\end{equation*}
Let $D=p_1$, this is equivalent to 
\begin{equation*}
\begin{aligned}
    \dot{u}_i &= D(u_{i-1} + u_{i+1} - 2 u_i) + (\gamma(v_{i+1}) - \gamma(v_i))(u_{i+1} + u_i)- (\gamma(v_{i}) - \gamma(v_{i-1}))(u_i + u_{i-1}).
    \end{aligned}
\end{equation*}
The corresponding differential equation is
\begin{equation*}
    u_t = D \Delta u + \nabla \cdot ( u \nabla \gamma(v)) = D \Delta u + \nabla \cdot ( u \nabla \gamma(g *_t u)). 
\end{equation*} If $\gamma$ is linear, it includes the previous case. This is a diffusion advection type equation. From the derivation, it is evident that individuals navigate not by directly using the population density but by comparing the differences between the accumulated local population information and that of nearby regions. 
\subsubsection{Fickian type diffusion}
We assume that animals can adjust their movement strategies and match their dispersal rate with individuals in the vicinity based on their memory of social behavior. If the dispersal rates are identical between adjacent patches, i.e., \( c_{ij} = c_{ji} \), and let \( c_{i i+1} = c_{i+1 i}=:\gamma_{i+1/2}\), we can derive a Fickian diffusion model.  
For each \( i \)-th patch, individuals measure the historical population \( v_i \). When individuals from either patch \( i \) or patch \( i+1 \) move, they adopt the same dispersal rate \( \gamma_{i+1/2}=\gamma(v_{i+1/2}) \), determined by historical population information from the \( i \)-th and \( i+1 \)-th patches. For example
$\gamma(v_{i+1/2})= \gamma(\theta v_i +(1-\theta)v_{i+1})$ for some $\theta\in[0,1]$. Then it follows that
\begin{equation*}
    \dot{u}_i = \gamma(v_{i+1/2})(u_{i+1} - u_i) - \gamma(v_{i-1/2})(u_i - u_{i-1}).
\end{equation*}
The corresponding continuous differential equation is 
\begin{equation*}
    u_t = \nabla\cdot (\gamma(v)\nabla u) = \nabla\cdot \left(\gamma\left(g *_t u \right)\nabla u\right).
\end{equation*} 
Notice that this model uses solely a quantity of memory, while a diffusion advection case requires a gradient of memory. By matching the dispersal rate over the vicinity, the whole movement structure loses its biased movement from the macroscopic model.
\subsubsection{Fokker-Planck type diffusion}
If animals decide their movement based solely on local temporal information, and the departing probabilities are identical for each patch, i.e. $c_{i-1i} = c_{i+1i}$.
Denote the dispersal probabilities as \( c_{i-1i} = c_{i+1i} = \gamma(v_i) \), the population dynamics at patch \( i \) are described by
\begin{equation*}
    \dot{u}_i = \gamma(v_{i+1}) u_{i+1} - 2\gamma(v_i) u_i + \gamma(v_{i-1}) u_{i-1}.
\end{equation*}
The corresponding continuous form of this equation is
\begin{equation*}
    u_t =\Delta (\gamma(v) u) =\Delta (\gamma(g *_t u) u).
\end{equation*}
This Fokker-Planck diffusion model uses a quantity of memory.
In the Fickian diffusion case, animals match their dispersal rates with adjacent patches, resulting in symmetric movement. However, in this model, naturally biased movement becomes evident when the equation is expanded:
\begin{equation*}
    u_t = \Delta (\gamma(g *_t u) u) = \nabla \cdot \left(\nabla \gamma(g *_t u) u + \gamma(g *_t u) \nabla u\right),
\end{equation*}
as it accounts for the gradients of historical population information.

\subsection{Discrete velocity jump process in one-dimensional space}
Random-walk models with heterogeneity have been extensively analyzed to derive specific diffusion operators \citep{YJK2019, HYkim2021, Othmer2004}, and their approximations have long been foundational for describing responses of biological organisms to landscape heterogeneity \citep{Hillen2002}. These models account for behavioral dynamics such as changes in movement speed and turning frequencies \citep{kareiva1987swarms, turchin1991translating}.
To be specific, the system
\begin{equation}
\label{eq: simple transport equation}
    \begin{aligned}
        \frac{\partial u^+}{\partial t} + c \frac{\partial u^+}{\partial x} & = - \lambda^+ u^+ + \lambda^- u^-, \\
        \frac{\partial u^-}{\partial t} - c \frac{\partial u^-}{\partial x} & =  \lambda^+ u^+ - \lambda^- u^-.
\end{aligned}
\end{equation} demonstrates the discrete velocity jump process in one-dimensional space. \( u^+ \) denote the population moving to the right and \( u^- \) denote the population moving to the left. The total population is \( u^+ + u^- \). It is assumed that individuals move at a constant speed \( c \). The symbols \( \lambda^+ \) and \( \lambda^- \) represent the turning frequency for individuals initially moving to the right (left) and subsequently turning to the left (right), respectively.
In the simplest case, where individuals lack the ability to use information from their surroundings, then the turning frequencies \( \lambda^+ \) and \( \lambda^- \) are constants and their movement is a random walk on a homogeneous environment, such as heat diffusion. However, most animals are able to adjust their movement based on social group information, which is typically accumulated over time. This process is modeled using the convolution \( g *_t u\) with the weighting function \( g(t) \) since we consider the continuous time domain.

\subsubsection{Diffusion-Advection model}

 We assume that $u^+$ and $u^-$ are mechanistic particles that can change their direction using an explicit memory $g*_tu$ with a rate function $\gamma$. To be specific, they are assumed to have an ability to measure and use the gradient $\gamma(g*_tu)_x$ when they decide directions. For the aggregation situation, we take a decreasing $\gamma$, while an increasing $\gamma$ is taken for the repulsion case.
This behavior modifies the transport equations \eqref{eq: simple transport equation} as follows:
\begin{align}
\label{eq: transport 1}
        \frac{\partial u^+}{\partial t} + \frac{c}{\varepsilon} \frac{\partial u^+}{\partial x} & = - \frac{1}{\varepsilon^2}\big(\lambda_1 + \lambda_2 \gamma(g *_t u)_x\big) u^+ + \frac{1}{\varepsilon^2} \big(\lambda_1 - \lambda_2 \gamma(g *_t u)_x\big) u^-, \\
\label{eq: transport 2}
        \frac{\partial u^-}{\partial t} - \frac{c}{\varepsilon} \frac{\partial u^-}{\partial x} & =  \frac{1}{\varepsilon^2} \big(\lambda_1 + \lambda_2 \gamma(g *_t u)_x\big) u^+ - \frac{1}{\varepsilon^2} \big(\lambda_1 - \lambda_2 \gamma(g *_t u)_x\big) u^-,
\end{align}
where \( \lambda_1>0 \) represents the constant tumbling frequency, and \( \lambda_2 \gamma(g *_t u)_x \) denotes the tumbling frequency depending on the gradient $\gamma(g *_t u)_x $ together with a constant $\lambda_2>0$, influenced by all accumulated information from past time. 
Some research has also studied particle movements influenced by turning frequencies. \citep{Othmer2004} proposed a turning frequency containing the evolution of the intracellular state. \citep{eftimie2007complex} discussed animal communication mechanisms based on attraction, repulsion, and alignment.
In our case, a particle moves and reverses its direction according to a Poisson process with the turning frequency 
\begin{equation*}
    \lambda^\pm = \lambda_1 \pm \lambda_2 \gamma(g *_t u)_x.
\end{equation*} 
Denote
\begin{equation*}
    \begin{aligned}
        u &= u^+ + u^-, \\
        j &= u^+ - u^-.
    \end{aligned}
\end{equation*}
Adding equations \eqref{eq: transport 1} and \eqref{eq: transport 2} implies
\begin{equation}
\label{eq:transport3}
     \frac{\partial u }{\partial t} + \frac{c}{\varepsilon} \frac{\partial j}{\partial x} = 0.
\end{equation}
Taking the derivative with respect to \( t \), we obtain
\begin{equation}
     \label{eq: transport 4}
     \frac{\partial^2 u}{\partial t^2} + \frac{c}{\varepsilon} \frac{\partial^2 j}{\partial t \partial x} = 0.
\end{equation}
Subtracting \eqref{eq: transport 2} from \eqref{eq: transport 1} implies
\begin{equation*}
     \begin{aligned}
     \frac{\partial j }{\partial t} + \frac{c}{\varepsilon} \frac{\partial u}{\partial x} & = -\frac{2 \lambda_1}{\varepsilon^2} j - \frac{2 \lambda_2}{\varepsilon^2}  u\gamma(g*u)_x.
     \end{aligned}
\end{equation*}
Taking the derivative with respect to \( x \), we obtain
\begin{equation}
     \label{eq: transport 5}
     \begin{aligned}
     \frac{\partial^2 j }{\partial x\partial t} + \frac{c}{\varepsilon} \frac{\partial^2 u}{\partial x^2} = - \frac{2 \lambda_1}{\varepsilon^2} \frac{\partial j}{\partial x} - \frac{2 \lambda_2}{\varepsilon^2} \frac{\partial (u\gamma(g*u)_x)}{\partial x}.
     \end{aligned}
\end{equation}
Combining \eqref{eq: transport 4} and \eqref{eq: transport 5}, we find u satisfies the following wave equation
\begin{equation*}
    \begin{aligned}
     \frac{\partial^2 u}{\partial t^2} - \frac{c^2}{\varepsilon^2} \frac{\partial^2 u}{\partial x^2} & = - \frac{2 \lambda_1}{\varepsilon^3} \frac{\partial u}{\partial t} + \frac{2 \lambda_2 c }{\varepsilon^3} \frac{\partial (u\gamma(g*u)_x)}{\partial x},
     \end{aligned}
\end{equation*}
where \( c \frac{\partial j}{\partial x} = - \frac{\partial u}{\partial t} \) \eqref{eq:transport3} is used.
It follows that
\begin{equation*}
\label{eq: }
    \begin{aligned}
     \frac{\varepsilon^3}{2 \lambda_1}\frac{\partial^2 u}{\partial t^2} +  \frac{\partial u}{\partial t}  &= \frac{\varepsilon c^2 }{2 \lambda_1} \frac{\partial^2 u}{\partial x^2}  + \frac{\lambda_2 c}{\lambda_1}   \frac{\partial (u\gamma(g*u)_x)}{\partial x}.
     \end{aligned}
\end{equation*}
Let \( \lambda_1 \to 0 \), \( \lambda_2 \to 0 \), \( \varepsilon \to 0 \) and suppose that \( \frac{\varepsilon c^2 }{2 \lambda_1} \to D\sim O(1) \) and \( \frac{\lambda_2 }{\lambda_1}\to K \sim O(1) \), where $K>0$ is some constant. Assume that $u = u^++u^-$ is bounded, then the equation is reduced to
\begin{equation*}
    \begin{aligned}
     \frac{\partial u}{\partial t}  &= D \frac{\partial^2 u}{\partial x^2}  +\frac{\partial (u\gamma(g*u)_x)}{\partial x},
     \end{aligned}
\end{equation*}
where the constant \( cK\) is absorbed into the function \( \gamma \). The corresponding higher-dimensional version is 
\begin{equation*}
    \begin{aligned}
     \frac{\partial u}{\partial t}  &= D \Delta u + \nabla \cdot \left(u\nabla \gamma\left(g *_t u\right)\right).
     \end{aligned}
\end{equation*}
The first term \( D\Delta u \) represents random diffusion, and the second term \( \nabla \cdot \left(u\nabla \gamma\left(g *_t u\right)\right) \) represents directed movement influenced by the gradient of temporal information collection.




\subsubsection{Fickian type and Fokker-planck type diffusion models}
Following the method in \cite{HYkim2021}, we consider a model with the tumbling frequency $\lambda(g*_tu)$ and speed $v(g*_tu)$:
\begin{align*}
        \frac{\partial u^+}{\partial t} +  \frac{1}{\varepsilon}\frac{\partial }{\partial x} (v(g*_tu)u^+)& = -\frac{1}{2\varepsilon^2} \lambda(g *_t u)( u^+ -u^-), \\
        \frac{\partial u^-}{\partial t} - \frac{1}{\varepsilon}\frac{\partial}{\partial x} (v(g*_tu)u^-) & =  \frac{1}{2\varepsilon^2} \lambda(g *_t u)( u^+ -u^-).
\end{align*} The parameter $\eps$ is introduced for parabolic scaling limit. This is because we want to see a diffusion-scale movement behavior through the transport equations.  Note that the speed $v$ and the tumbling frequency $\lambda$ are determined by the memory quantity $g*_tu$, not its gradient.
As we did in the diffusion advection case, by denoting the flux $j=\frac{v(u^+-u^-)}{\varepsilon}$, we obtain 

\begin{equation*}
\begin{aligned}
    &\frac{\partial u}{\partial t} = -\frac{\partial}{\partial x}j, \text{ and}\\
    &\frac{\varepsilon^2}{v(g*_tu)}\frac{\partial j}{\partial t} +\varepsilon^2j\frac{\partial}{\partial t}\left(\frac{1}{v(g*_tu)}\right) + \frac{\partial}{\partial x}(v(g*_tu)u) = -\frac{\lambda(g*_tu)}{v(g*_tu)}j.
\end{aligned}
\end{equation*} If we assume that $\gamma$ and $v$ are bounded away from zero and bounded above, and assume that $u^i$ are bounded, one can easily prove that $\frac{\varepsilon^2}{v(g*_tu)}\frac{\partial j}{\partial t} +\varepsilon^2j\frac{\partial}{\partial t}\left(\frac{1}{v(g*_tu)}\right)\to 0$. Then we can derive the following equation
\begin{equation*}
    \frac{\partial u}{\partial t} = \frac{\partial}{\partial x}\left(\frac{v(g*_tu)}{\lambda(g*_tu)}\frac{\partial}{\partial x}\left(v(g*_tu)u\right)\right).
\end{equation*} If $v\equiv 1$ and $\lambda(s) =\gamma(s)^{-1}$, it reads a Fickian type diffusion
\begin{equation*}
    \frac{\partial u}{\partial t} = \frac{\partial}{\partial x}\left(\gamma(g*_tu)\frac{\partial u}{\partial x}\right).
\end{equation*} If $v(s) = \gamma(s)$ and $\lambda(s) = \gamma(s)$, then it reads a Fokker-Planck type diffusion 
\begin{equation*}
    \frac{\partial u}{\partial t} = \frac{\partial^2}{\partial x^2}\left(\gamma(g*_tu)u\right).
\end{equation*}
The Fickian type considers the tumbling frequency dependent on memory, while the movement speed remains independent. In contrast, the Fokker-Planck type accounts for both the tumbling rate and the movement speed being influenced by memory.
If $\gamma$ increases, the speed and the tumbling frequency also increase.

\section{Numerical simulation}
\label{sec:Numerical simulation}
We now conduct numerical simulations to visualize the influence of temporal information on animal movement and observe differences between three different diffusion mechanisms with explicit memory. 

The temporal weighting function \( g(t) \) characterizes the dependence of memory on past information. Here, we consider the Gamma distribution. Specifically, the Gamma distribution function of order \( k \) (where \( k \in \mathbb{N} \cup \{0\} \)) is defined as
\begin{equation}
\label{eq: gamma kernel}
g_k(t) = \frac{t^k e^{-t/\tau}}{\tau^{k+1} \Gamma(k+1)}.
\end{equation}
Two special cases of the Gamma kernel, commonly analyzed, are referred to as the weak kernel and the strong kernel \citep{shi2021spatial}. When \( k = 0 \), the kernel is termed the weak kernel. This kernel is strictly decreasing with \( t \), reflecting the biological phenomenon of fading memory over time, where older information becomes less significant. In contrast, when \( k > 0 \), the kernel is called the strong kernel. The maximum extremum occurs at \( k\tau \). This kernel initially increases over \( [0, k\tau) \) and then decreases over \( (k\tau, \infty) \), representing two key biological phases: knowledge acquisition followed by memory decay, as illustrated in Figure \ref{fig:gamma}.
The parameters \( k \) and \( \tau \) determine the average delay and the extent of memory, which are critical to quantifying the effects of spatial memory on dynamics. The mean and variance of \( g_k(t) \) are given by
\[
\mathbb{E}(g_k(t)) = (k+1)\tau, \quad \text{and} \quad \mathbb{V}\mathrm{ar}(g_k(t)) = (k+1)\tau^2.
\]
As the parameters \( k \) or \( \tau \) increase, the temporal kernel \( g_k(t) \) emphasizes information from earlier time points and encompasses a longer duration of past information, as shown in Figure \ref{fig:gamma}. For the subsequent analysis, we set \( \tau = 0.5 \) and examine both the weak kernel (\( k = 0 \)) and the strong kernel (use $k=5$ as an example). 

\begin{figure}[h!]
    \centering
    \begin{subfigure}[t]{0.5\linewidth}
        \centering
        \includegraphics[width=\linewidth]{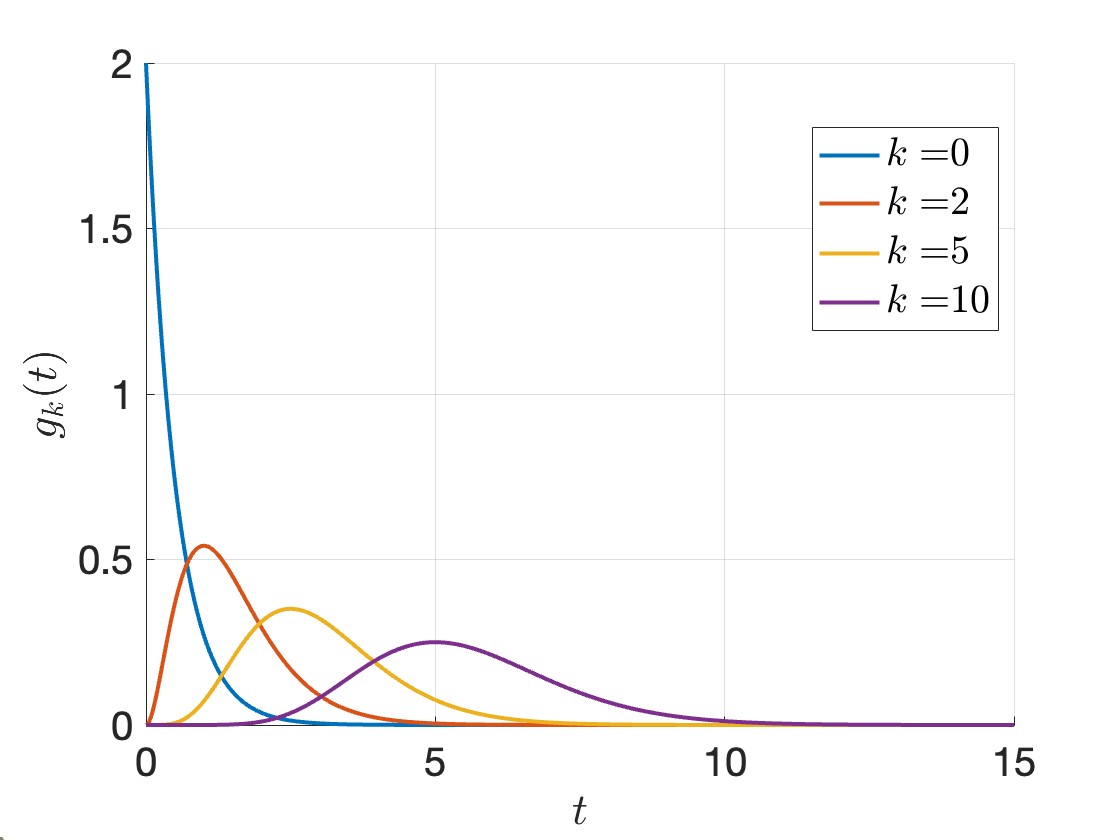}
        \caption{}
        \label{fig:gamma_k.jpg}
    \end{subfigure}%
        \begin{subfigure}[t]{0.5\linewidth}
        \centering
        \includegraphics[width=\linewidth]{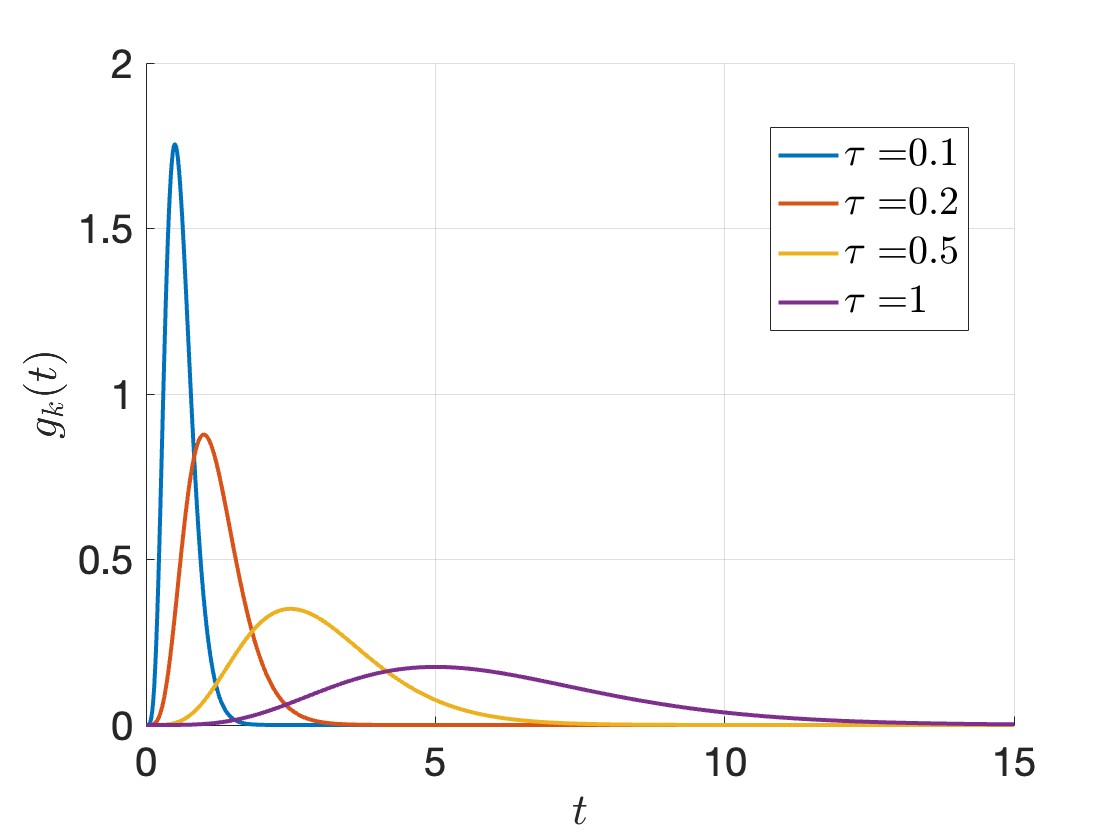}
        \caption{}
        \label{fig:gamma_tau.jpg}
    \end{subfigure}%
    \caption{Temporal kernel functions \( g_k(t) \) over the interval \([0,15]\): (a) varying \( k \) with fixed \( \tau = 0.5 \); (b) varying \( \tau \) with fixed \( k = 5 \).}
    \label{fig:gamma}
\end{figure}

We now investigate the movement driven by temporal distributed memory. The governing equations are given as
\begin{equation}
\label{eq: equations for numerics}
\begin{aligned}
    &u_t = Du_{xx} + (u \gamma\left(g*_t u\right)_x)_x,\\
    &u_t = (\gamma(g*_t u)u_x)_x,\\
    &u_t = (\gamma(g*_t u)u)_{xx},
\end{aligned}
\end{equation}
where \( g \) is given in equation \eqref{eq: gamma kernel}. 
Those equations are the diffusion advection model, the Fickian diffusion model, Fokker-Planck diffusion model, respectively. If we rewrite the equations in the general form, we have
\begin{equation*}\label{eq: oned generalized}
    u_t = (D(g*_tu)u_x + A(g*_tu)u(g*_tu)_x)_x.
\end{equation*}
We take the homogeneous Neumann boundary condition for numerical simulations to accurately capture the dispersal behavior. The well-posedness of the generalized equation with the corresponding boundary condition, presented in Appendix \ref{app:wellposedness}, serves to validate the three derived models \eqref{eq: equations for numerics} mathematically.

We consider two different scenarios.
\begin{itemize}
    \item [(1)] (Memory-induced repulsion) The case where $A(s)>0$.  Individuals avoid the crowded region based on their accumulated memory. For the three examples \eqref{eq: equations for numerics}, it is the case when $\gamma$ is a increasing function. 
    \item [(2)] (Memory-induced attraction) The case where $A(s)<0$. Individuals move toward the crowded region based on their accumulated memory. For the three examples \eqref{eq: equations for numerics}, it is the case when $\gamma$ is a decreasing function. 
\end{itemize}

  \subsection{The case of memory-induced repulsion: $A(s)>0$}
  In the case when $A(s)>0$ with the weak kernel $g(t) = \tau^{-1}e^{-t/\tau}$, any constant steady-states $\bar u$ is linearly stable from Lemma \ref{lem: stability}. Indeed,
  \[\chi(\bar u) = -\frac{A(\bar u) \bar u}{D(\bar u)}<0. \] Even if we consider a strong kernel, we choose the same $D$, $A$, and $\bar u$ such that $\chi$ has the same value as in the weak kernel case.
  
  Let an initial function be given by
  \[ u(x,0) = 0.3*\sin(\pi(x+0.5)) + 1.5,\] 
so its average is $\bar u= \frac{1}{|\Omega|}\int_\Omega u_0dx=1.5$. Since we want to see the effect of temporal memory and its result on an individual's directed movement, we consider the fixed ratio $|\chi|\approx 2$ for all three cases, so the advection occurs at a faster rate than the diffusion.
Therefore, we take $D$ and $A$ for the ratio $\chi$ around the average of the initial function $\bar u = 1.5$ to be approximately near 2.
For the diffusion advection equation, we take $D=1.5, \alpha=2$, so $|\chi|=\frac{a}{D}\bar u = 2$. For both the Fickian type and Fokker-Planck type diffusion, we take $\gamma(s) = 0.5(s+0.05)^2+0.02$. In the case of Fokker-Planck case, we have $|\chi(\bar u)|=\gamma'(\bar u)\bar u/\gamma(\bar u)\approx 1.9$. We conduct simulations of equations using two distinct Gamma distributions: a weak kernel (\( k = 0 \)) and a strong kernel (\( k = 5 \)). In the weak kernel scenario, systems with all three diffusion types rapidly converge to a uniform stable state; the Fokker Planck diffusion exhibits a slight delay but also transitions to uniformity swiftly (see Figure \ref{fig:positiveA_k0}). In the strong kernel case, we observe time delay phenomena with a wiggle pattern at the beginning that gradually fades away, in all models except for the Fickian diffusion model.


  \begin{figure}[h!]
    \centering
    \begin{subfigure}[t]{0.33\linewidth}
        \centering
        \includegraphics[width=\linewidth]{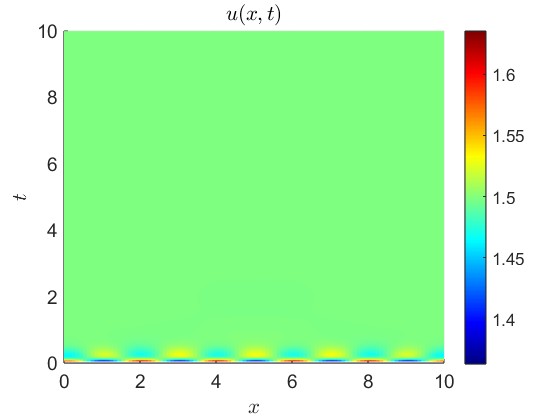}
        \caption{}
        \label{fig:diffadv_nosource_k0.jpg}
    \end{subfigure}%
\hfill
    \begin{subfigure}[t]{0.33\linewidth}
        \centering
        \includegraphics[width=\linewidth]{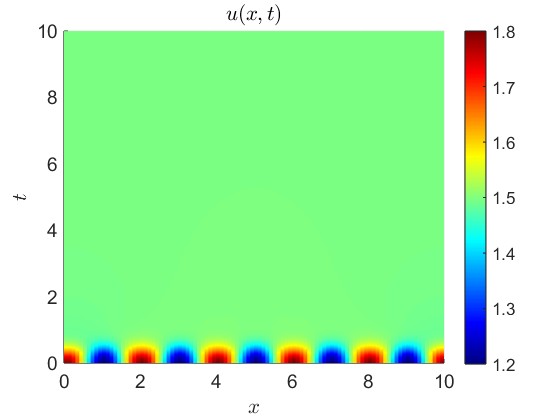}
        \caption{}
        \label{fig:Fick_nosource_k0.jpg}
    \end{subfigure}%
    \hfill
    \begin{subfigure}[t]{0.33\linewidth}
        \centering
        \includegraphics[width=\linewidth]{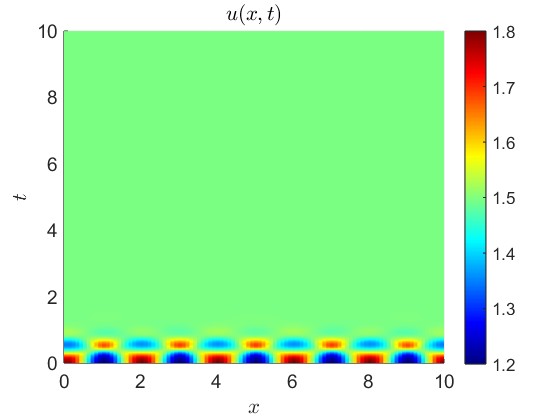}
        \caption{}
        \label{fig:FP_nosource_k0.jpg}
    \end{subfigure}
    
\caption{Spatial temporal distribution of population density with a weak memory kernel (\(k=0\)) for $A(s) >0$ case: (a) diffusion advection type; (b) Fickian type diffusion; (c) Fokker Planck type diffusion.}
    \label{fig:positiveA_k0}
    \end{figure}
    

\begin{figure}[h!]
    \centering
    \begin{subfigure}[t]{0.33\linewidth}
        \centering
        \includegraphics[width=\linewidth]{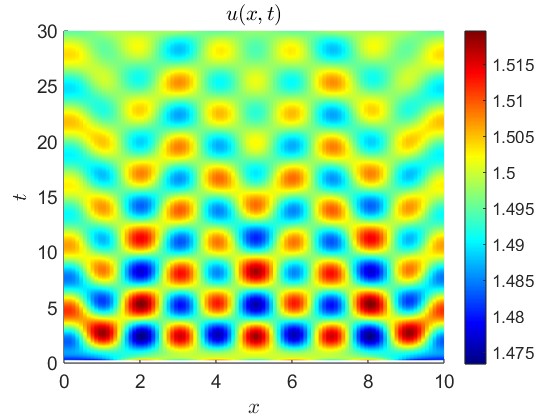}
        \caption{}
        \label{fig:diffadv_nosource_k5.jpg}
    \end{subfigure}%
\hfill
    \begin{subfigure}[t]{0.33\linewidth}
        \centering
        \includegraphics[width=\linewidth]{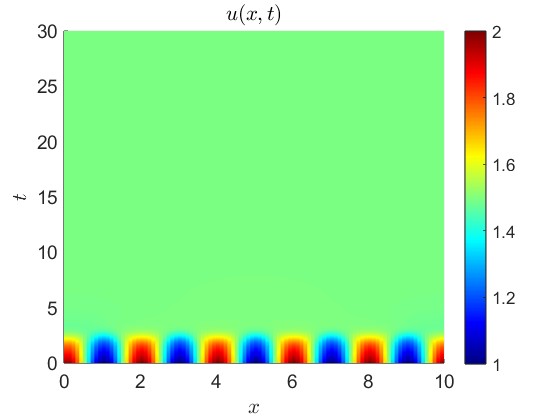}
        \caption{}
        \label{fig:Fick_nosource_k5.jpg}
    \end{subfigure}%
    \hfill
    \begin{subfigure}[t]{0.33\linewidth}
        \centering
        \includegraphics[width=\linewidth]{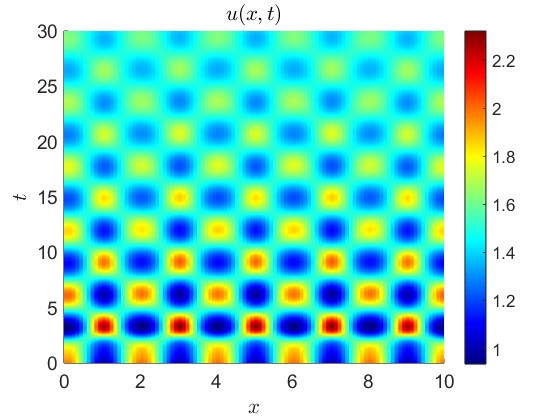}
        \caption{}
        \label{fig:FP_nosource_k5.jpg}
    \end{subfigure}
    
\caption{Spatial temporal distribution of population density with a strong memory kernel (\(k=5\)) for $A(s) >0$ case: (a) wiggle pattern for the diffusion advection equation; (b) Fickian type diffusion; (c) wiggle pattern for Fokker-Planck type diffusion.}
    \label{fig:positiveA_k5}
    \end{figure}
    
\subsection{The case of memory-induced attraction: $A(s)<0$}
In the case when $A(s)<0$ with a weak kernel, as we can see from Lemma \ref{lem: stability} it is possible that $\chi(\bar u) >1$ depending on the choice of parameters, so we expect aggregation phenomena. As they use accumulated information to move toward high-density regions, we expect that they form aggregates. 

In this case, the kernel function 
\(g\) influences only the speed of aggregate formation.
In this subsection, we take the weak kernel $g(t) = \tau^{-1}e^{-t/\tau}$ with $\tau =0.5$, which allows us to use Lemma \ref{lem: stability} for choice of parameters for each equation based on linear stability results (see Appendix \ref{app:linear stability}).

Let an initial function be given by
  \[ u(x,0) = 0.3*\sin(\pi(x+0.5)) + 1,\] 
so $\bar u =1$. As we did in the memory-dependent repulsion case, we take the ratio $\chi\approx 1.3$.
For the diffusion advection equation, we take $D=1.5, \alpha=-2$, so $\chi=-\frac{\alpha}{D}\bar u = 1.3$. For both the Fickian type and Fokker-Planck type diffusion, we take $\gamma(s) = \frac
{0.5}{(s+0.5)^2}+0.01$. In the case of Fokker-Planck diffusion, we have $\chi(\bar u)=-\gamma'(\bar u)\bar u/\gamma(\bar u)\approx 1.3$.

\begin{figure}[h!]
    \centering
    \begin{subfigure}[t]{0.33\linewidth}
        \centering
        \includegraphics[width=\linewidth]{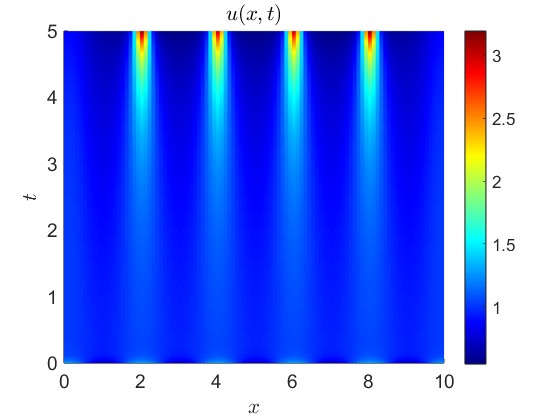}
        \caption{}
        \label{fig:diffadv_nosource_k0_blowup.jpg}
    \end{subfigure}%
\hfill
    \begin{subfigure}[t]{0.33\linewidth}
        \centering
        \includegraphics[width=\linewidth]{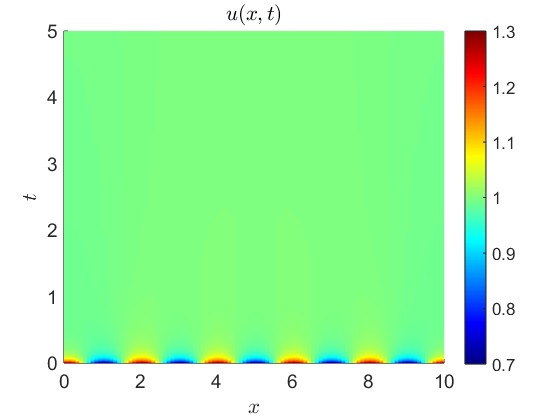}
        \caption{}
        \label{fig:Fick_nosource_k0_aggregation.jpg}
    \end{subfigure}%
    \hfill
    \begin{subfigure}[t]{0.33\linewidth}
        \centering
        \includegraphics[width=\linewidth]{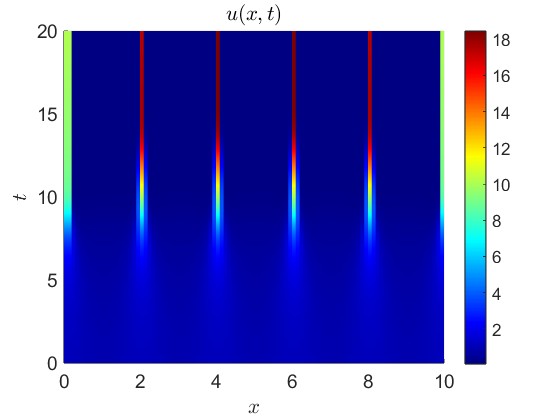}
        \caption{}
        \label{fig:FP_nosource_k0_aggregation.jpgb}
    \end{subfigure}
    
\caption{Spatial temporal distribution of population density for the case \( A(s) < 0 \): (a) diffusion advection type; (b) Fickian type diffusion; (c) Fokker Planck type diffusion. The peaks of the initial function grow over time, while the total population remains conserved.}
    \label{fig:singleaggregation}
    \end{figure}
    
As illustrated in Figure \ref{fig:singleaggregation}, collective behavior is observed from the diffusion advection equation and the Fokker-Planck type diffusion model because there is no advection term induced by temporal memory about population density.

\section{Discussion}
Memory plays a crucial role in determining where and how animals move and adjusting their dispersal strategies over time. Diffusion models are a key tool in population ecology for studying animal movement. In this study, we investigate how the use of explicit memory in different ecological dispersal scenarios leads to distinct mathematical models. 

The resulting models correspond to three types of diffusion: diffusion advection, Fickian diffusion, and Fokker-Planck diffusion, each representing a distinct movement strategy: gradient-based movement, environment matching, and location-based movement, respectively. Specifically, when animals rely on the memory of spatial gradients from historical population data, their movement is captured by the advection-diffusion equation. Matching their dispersal rate to that of individuals in the surrounding environment results in Fickian diffusion. Finally, when movement decisions are based solely on the memory of local historical information, with a departure rate determined by this local memory, the behavior aligns with Fokker-Planck diffusion.

These models are systematically derived by using three approaches: discrete velocity jump processes, time-space discretization, and patch-based models. 
These derivations highlight the connections between ecological dispersal strategies and their mathematical formulations.

Driven by distinct mechanisms, animals may avoid crowded regions based on accumulated memory or move toward them \citep{ford2013two,matthysen2005density,cressman2011effects}. We numerically analyzed these behaviors using three commonly employed diffusion models: diffusion advection, Fickian, and
 Fokker Planck types. In memory-dependent repulsion with a weak kernel, all models converged to a stable state, with the Fokker Planck model displaying a slight delay before
 achieving uniformity. For a strong kernel, wiggle patterns were observed in all models,
 while the Fokker Planck model demonstrated the most pronounced pattern. In memory-dependent attraction, distinct collective behaviors emerged except the Fickian type, with
 varying aggregation speeds across the models and kernels. 
 
To validate the derived framework, we also mathematically analyzed the existence of the solution to a general animal movement model \eqref{eq: general model no f} that unifies our three movement models. Additionally, we explored the possibility of finite-time blow-up phenomena by considering a special case with a weak memory kernel. Existing studies on accumulated explicit memory models are limited, and this work addresses a critical gap in the literature, as previous memory models primarily focused on equilibrium properties without exploring solution existence \citep{shi2021spatial,ji2024stability}.



This framework can be readily extended to more complex scenarios. For example, incorporating population dynamics into \(\eqref{eq: main equation}\) yields:  
\[
u_t = \nabla \cdot \left(D(g *_t u)\nabla u + A(g *_t u) u \nabla(g *_t u)\right) + f(u),
\]  
where \(f(u)\) is a reaction term. For instance, with intraspecific competition, it can be \(f(u) = u(1 - u)\). Additionally, if movement is influenced by a spatially and temporally heterogeneous source \(m(x,t)\), \(\eqref{eq: main equation}\) becomes:  
\[
u_t = \nabla \cdot \left(D(g *_t m)\nabla u + A(g *_t m) u \nabla(g *_t m)\right) + f(u,m),
\]  
where \(m(x,t)\) varies across space and time. By adjusting the signs of \(D\) and \(A\), the model can represent attraction or repulsion to various sources \(m\), such as food or toxins. The framework can accommodate multiple sources and interactions between multiple species. Although this study primarily focuses on one-dimensional space and temporal memory, higher spatial dimensions and spatial memory could be explored in future work to capture more complex ecological phenomena.


\section*{Acknowledgments}
The research of Hao Wang was partially supported by the Natural Sciences
and Engineering Research Council of Canada (Individual Discovery Grant RGPIN-2020-03911 and Discovery
Accelerator Supplement Award RGPAS-2020-00090) and the Canada Research Chairs program (Tier 1 Canada
Research Chair Award).

\appendix

\section{Well-posedness}
\label{app:wellposedness}

In this section, we prove the existence of the solution for a general model in one-dimensional space,
\begin{equation}
    \label{eq: general model no f}
    \left\{\begin{aligned}
        &u_t = \left(D(g*_t u)u_{x} + A(g*_t u)u \left(g*_t u\right)_x\right)_x \quad (x,t) \in [0,L]\times(0,T_{\max}]\\
        &  u_x(0,t)=u_x(L,t) = 0, \quad  t \in (0,T_{\max}]\\
        & u(x,0) = u_0(x) ,\quad x \in [0,L]\\
        &g*_t u = \int_0^t g(t-s)u(x,s)\, \mathrm{d}s,
    \end{aligned}\right.
\end{equation} 
By choosing functions \(D\) and \(A\), we can obtain all three diffusion types, as discussed in the Introduction section. 
For the analysis in this chapter, we always assume the following conditions on $D,A$ and $g$:
\begin{itemize}
    \item [\textbf{(A1)}] $D$, $A \in C^3[0, \infty)$ and $\inf_{s \geq 0} D(s) > 0$. 
    \item [\textbf{(A2)}] $g \in C [0,\infty) \cap L^1[0,\infty)$.
\end{itemize}

Denote the space \( X_{\alpha,T} = C^{\alpha, \frac{\alpha}{2}}(\bar{\Omega} \times [0, T]) \), for any \( \alpha \geq 0 \) and \( T > 0 \), with the corresponding seminorm \( [\cdot]_{\alpha,T} \) and norm \( \|\cdot\|_{\alpha,T} \). It follows that
\( X_{0,T} = C^0(\bar{\Omega} \times [0, T]) \) with the norm $
\| \cdot \|_{0,T} = \| \cdot \|_{C(\bar{\Omega} \times [0, T])}$. For definitions of these norms, please refer (Section 4.1 \citep{lieberman1996second})


\begin{lemma}
    \label{lem: g*v gammaT}
    If $g \in C([0,T])$ and $v \in X_{\gamma,T}$ for some $\gamma \in [0,1)$, then
    \begin{equation*}
        \left\|g *_t v\right\|_{\gamma,T} \leq  2T \left\|g\right\|_{C([0, T])} \left\|v\right\|_{\gamma,T}.
    \end{equation*}
\end{lemma}
\begin{proof}
    The case $\alpha=0$ is trivial. Now consider $0<\gamma<1$, for any $x,y \in \bar{\Omega}$, and $t,s \in [0,T]$, $t\ge s$, we have 
    \begin{equation*}
        \begin{aligned}
            &\quad \left| \left(g*_t v\right)(x,t) - (g *_s v)(y,s)\right|\\
            & = \left|\int_0^t g(r)v(x,t-r) \mathrm{d}r - \int_0^s g(r) v(y,s-r) \mathrm{d}r \right|\\
            & \leq \int_0^s |g(r)| \left|v(x,t-r) - v(y,s-r)\right|\mathrm{d}r + \int_0^{t-s} \left|g(t-r)\right|\left|v(x,r)\right| \mathrm{d}r\\
            & \leq \left\|g \right\|_{C([0,T])} \left( s [v]_{\gamma,T}\left(\left| x-y\right|^\gamma + \left|t-s\right|^{\frac{\gamma}{2}} \right) + (t-s) \left\|v\right\|_{0,T}\right)\\
            & \leq T \left\|g\right\|_{C([0,T])} \left\|v\right\|_{\gamma,T} \left(\left| x-y\right|^\gamma + \left|t-s\right|^{\frac{\gamma}{2}} \right). 
        \end{aligned}
    \end{equation*}
    That is 
    \begin{equation*}
        [g*_t v]_{\gamma,T} \leq T \left\|g\right\|_{C([0,T])} \left\|v\right\|_{\gamma,T}.
    \end{equation*}
    It follows that
    \begin{equation*}
        \left\|g*_t v\right\|_{\gamma,T} = \left\|g*_t v\right\|_{C([0,T])} + [g*_t v]_{\gamma,T}  \leq 2T \left\|g\right\|_{C([0,T])} \left\|v\right\|_{\gamma,T}.
    \end{equation*}
\end{proof}


We are ready to prove the local existence of \eqref{eq: general model no f}.
\begin{theorem}[Local existence]\label{thm:existence}
    Let $0<\gamma<1$, and suppose that $(A1)-(A2)$ hold. If $u_0\in C^{2+\gamma}[0,L]$ and $u_0\ge 0$ in $(0,L)$, then there exists $T_{\max}>0$ and unique nonnegative $u\in C^{2+\gamma,1+\gamma/2}([0,T_{\max}]\times[0,L])$ that  solves the system \eqref{eq: general model no f}, and satisfies
    \begin{equation}\label{eq: mass conservation}
        \int_\Omega u(\cdot,t) dx = \int_\Omega u_0dx
    \end{equation} for all $t\in [0,T_{\max}]$.
    In addition, we have 
    \begin{equation}\label{blowup}
    \text{ either } T_{\max}=\infty \text{ or } 
        \|u\|_{2+\gamma,t}\to \infty \text{ as } t\to T_{\max} \text{ for some } \gamma \in (0, 1).
    \end{equation} 
\end{theorem}
\begin{proof}
    Fix $R \in \left( \left\|u_0\right\|_{C^{2+\gamma}(0,L)}, \infty \right)$ and $\gamma \in (0,1)$. Define 
    \begin{equation*}
        K_T = \left\{u \in X_{2+\gamma,T}: \left\|u\right\|_{2+\gamma,T} \leq R\right\}
    \end{equation*}
    For $v \in K_{\infty} = \underset{T\geq 0}{\cap} K_T$, consider the solution $u$ to the following equation 
    \begin{equation}\label{eq:auxiliary}\left\{
        \begin{aligned}
            &u_t  =  (D(g*_tv)u_{x})_x +  \left(A(g*_tv)u \left(g *_t v\right)_x \right)_x,  \quad (x,t)\in  (0,L) \times(0,T_{\max}]\\
        &  u_x(0,t)=u_x(L,t) = 0, \quad  t \in (0,T_{\max}]\\
        & u(x,0) = u_0(x) ,\quad x \in [0,L]\\
        \end{aligned}\right.
    \end{equation}
    By (Proposition 7.3.3 \citep{lunardi2012analytic}), there exists a unique classical solution defined on $[0,L]\times[0,T_1]$ for some $T_1>0$ independent of $v \in K_\infty$. By parabolic Schauder estimate (Theorem 4.31
 of \citep{lieberman1996second}), $\exists \; T_0 \in (0,T_1]$ sufficiently small such that $u \in K_{T_0}$. Define a map $S: K_{T_0} \rightarrow K_{T_0}$ by $S(v)=u$. $S$ can also be regarded as a map $K_T \rightarrow K_T$ for all $T \leq T_0$. 
 
 We claim that $S$ is a contraction in $K_T$ if $T$ is sufficiently small. Let $T \in(0,T_0]$ and take $v^{[1]}$, $v^{[2]} \in K_T$. Denote $u^{[i]}=S\left(v^{[i]}\right)$, $w^{[i]}= \left(g*_tv^{[i]}\right)_x$, $i=1,2$ and $\tilde{u} = u^{[1]} - u^{[2]}$, $\tilde{v} = v^{[1]} - v^{[2]}$, $\tilde{w} = w^{[1]} - w^{[2]}$. By calculation, $\tilde u$ satisfies the equation
 \begin{equation*}\left\{
     \begin{aligned}
         &\tilde u_t  =  (D(g*_tv^{[1]})\tilde u_{x})_x +\left(A(g *_t v^{[1]})w^{[1]}\tilde u\right)_x + F(x,t), &(x,t) \in (0,L)\times (0,T],\\
        &  u_x(0,t)=u_x(L,t) = 0, \quad  t \in (0,T_{\max}]\\
        & \tilde u(x,0) = 0 ,\quad x \in [0,L],\\
     \end{aligned}\right. 
 \end{equation*} where
 \begin{align*}
     &F(x,t) = \left((D(g*_tv^{[1]})-D(g*_tv^{[2]}))u^{[2]}_x\right)_x+\left(A(g*_tv^{[1]})u^{[2]}\tilde w \right)_x \\
     &\qquad \qquad +\left((A(g*_tv^{[1]})-A(g*_tv^{[2]}))u^{[2]}w^{[2]}\right)_x .
 \end{align*}

 By the Schauder estimate (Theorem 4.31
 of \citep{lieberman1996second}), there exists $C_1$ dependent on $L, T, M$, $\left\|D(g *_t v^{[1]})\right\|_{\gamma,T}$, $\left\|A(g *_t v^{[1]})w^{[1]}\right\|_{\gamma,T}$, and $\left\|\left(A(g *_t v^{[1]})w^{[1]}\right)_x\right\|_{\gamma, T}$ such that
 \begin{equation}
 \label{eq: u tilde}
 \begin{aligned}
     \left\|\tilde u \right\|_{2+\gamma,T} \leq C_1 \|F\|_{\gamma,T} .
     \end{aligned}
 \end{equation} 
By the \( C^{3} \) regularity of \( D \), we have
\[
\begin{aligned}
   \|\left((D(g *_t v^{[1]}) - D(g *_t v^{[2]}))u^{[2]}_x\right)_x\|_{\gamma, T} 
   & \leq \|(D(g *_t v^{[1]}) - D(g *_t v^{[2]}))u^{[2]}_x\|_{1+\gamma, T} \\
   & \leq \|D(g *_t v^{[1]}) - D(g *_t v^{[2]})\|_{1+\gamma, T} \|u^{[2]}_x\|_{1+\gamma, T} \\
   & \leq R \|D(g *_t v^{[1]}) - D(g *_t v^{[2]})\|_{1+\gamma, T} \\
   & \leq R \int_{g *_t v^{[2]}}^{g *_t v^{[1]}} \left\| D'(s) \right\|_{1+\gamma, T} \, \mathrm{d}s  \\
   & = R  \int_0^1 \left\| D'(v^s) (g *_t \tilde{v}) \right\|_{1+\gamma, T} \, \mathrm{d}s  \\
   & \leq R  \int_0^1 \left\| D'(v^s) \right\|_{1+\gamma, T} \, \mathrm{d}s  \|g *_t \tilde{v}\|_{1+\gamma, T},
\end{aligned}
\]
where \( v^s = s g *_t \tilde{v} + g *_t v^{[2]} \). Using Lemma \ref{lem: g*v gammaT}, we find
\[
\|g *_t \tilde{v}\|_{1+\gamma, T} = \|g *_t \tilde{v}\|_{C[0,T]} + \|g *_t \tilde{v}_x\|_{\gamma, T} \leq 3T \|g\|_{C[0,T]} \|\tilde{v}\|_{1+\gamma, T}.
\]
Moreover, for the integral term, by the boundedness of \( D'' \),
\[
\begin{aligned}
    \int_0^1 \left\| D'(v^s) \right\|_{1+\gamma, T} \, \mathrm{d}s 
    & = \int_0^1 \left\|D'(v^s) \right\|_{C[0,T]} \, \mathrm{d}s + \int_0^1 \left\| D''(v^s) v^s_x \right\|_{\gamma, T} \, \mathrm{d}s \\
    & \leq M + \int_0^1 \left\| D''(v^s) \, \right\|_{\gamma, T} \mathrm{d}s \|g *_t \tilde{v}_x + g *_t v^{[2]}_x\|_{\gamma, T}.
\end{aligned}
\]
Since \( v^{[2]} \) and \( \tilde{v} \) belong to \( X_{2+\gamma, T} \), by Lemma \ref{lem: g*v gammaT}, both \( \|g *_t \tilde{v}_x\|_{\gamma, T} \) and \( \|g *_t v^{[2]}_x\|_{\gamma, T} \) are bounded by \( 2T \|g\|_{C[0,T]} R \). Using (A1) regularity of \( D \), we deduce that \( D''(v^s) \in X_{\gamma, T} \) for $s \in [0,1]$. Therefore, there exists a constant \( c_1 > 0 \) such that
\[
\int_0^1 \left\| D'(v^s) \, \right\|_{1+\gamma, T} \mathrm{d}s \leq c_1.
\]
Combining the above, we conclude that
\[
\|\left((D(g *_t v^{[1]}) - D(g *_t v^{[2]}))u^{[2]}_x\right)_x\|_{\gamma, T} \leq 3c_1 RT \|g\|_{C[0,T]} \|\tilde{v}\|_{1+\gamma, T}.
\]
By a similar argument, for $A \in C^{3}[0,T]$, we can show  
  \begin{align*}
     \|\left((A(g*_tv^{[1]})-A(g*_tv^{[2]}))u^{[2]}w^{[2]}\right)_x\|_{\gamma,T}&\le 3c_2 RT \left\|g\right\|_{C[0,T]} \left\|\tilde v\right\|_{2+\gamma,T},
 \end{align*}
 for some constant $c_2$.
 By Lemma \ref{lem: g*v gammaT} and the condition (A1), we have 
  \begin{equation*}
     \begin{aligned}
         \quad \left\|\left(A(g*_tv^{[1]})u^{[2]}\tilde w\right)_x \right\|_{\gamma,T} 
         &\le  \left\|A(g*_tv^{[1]})u^{[2]}\tilde w \right\|_{1+\gamma,T}\\
         &\le  \left\|A(g*_tv^{[1]})\right\|_{1+\gamma,T} \left\| u^{[2]} \right\|_{1+\gamma,T} \left\| \tilde w \right\|_{1+\gamma,T}\\
 &\le 3c_3 R T \left\|g\right\|_{C[0,T]} \left\|\tilde v\right\|_{2+\gamma,T},
     \end{aligned}
 \end{equation*}
 where $c_3 = M + 2c_2  TR \left\|g\right\|_{C[0,T]} $. 
Similar estimates hold for $\left\|D(g *_t v^{[1]})\right\|_{\gamma,T}$, $\left\|A(g *_t v^{[1]})w^{[1]}\right\|_{\gamma,T}$, and $\left\|\left(A(g*_tv^{[1]})w^{[1]}\right)_x\right\|_{\gamma, T}$. Therefore, $C_1$ depends only on $L$, $T$, $M$, $\left\|g\right\|_{C([0,T])}$, $r$, and $R$.
    
 Now we continue with \eqref{eq: u tilde} to get for small enough $T>0$, we obtain
 \begin{equation*}
 \label{eq: u tilde 2}
     \begin{aligned}
         \left\|\tilde u \right\|_{2+\gamma,T} &\leq C_1 3 (c_1 + c_2 + c_3) RT \left\|g\right\|_{C[0,T]} \left\|\tilde v \right\|_{2+\gamma,T}. 
     \end{aligned}
 \end{equation*}
When we take small enough $T$, the constant before $\left\|\tilde v\right\|_{2+\gamma,T}$ can be less than 1. Therefore $S$ is a contraction when $T$ is small enough. By Banach fixed point theorem, there exists a unique $ u \in K_T$, such that $S(u)=u$. That is, $u$ is the unique solution in $[0,T]$. 
We deduce from \eqref{eq:auxiliary} and the nonnegative initial data $u_0$ with $u_0\not\equiv 0$ that the fixed point $u$ is nonnegative by using the maximum principle. 
By the standard argument, the solution $u$ can be extended up to some $T_{\max}>0$. Since the above choice of $T$ depends only on $\|u_0\|_{C^{2+\gamma}[0,L]}$, the conclusion \eqref{blowup} follows. Due to the homogeneous Neumann boundary condition on $u$, \eqref{eq: mass conservation} follows from $(g*_tu)_x = g*_t(u_x)$.
\end{proof}

\begin{remark}
    
    As we have seen from the existence result, it is possible to observe finite time blow-up from the equation \eqref{eq: general model no f}.
    For example, if we take $g(t)=\tau^{-1}e^{-t/\tau}$ and consider the system \eqref{eq: general model no f} becomes 
    \begin{equation}\label{eq: special case}
    \left\{\begin{alignedat}{2}
        &u_t = \nabla\cdot(D(v)\nabla u + A(v)u\nabla v) && \text{ in } \Omega \times (0,T),\\
        &\tau v_t = u - v,\\
        &\partial_\nu u = \partial_\nu v =0 && \text{ on } \partial\Omega \times [0,T],\\
        &u(x,0) = u_0(x), \, v(x,0) = 0 && \text{ in }\Omega.
    \end{alignedat}\right.
\end{equation} If we incorporate diffusion $\eps\Delta v$ in the second equation, it reads
\begin{equation}\label{eq: chemotaxis}
    \left\{\begin{alignedat}{2}
        &u_t = \nabla\cdot(D(v)\nabla u + A(v)u\nabla v), \\
        &\tau v_t = \eps\Delta v+ u - v.
    \end{alignedat}\right.
\end{equation} 
The system \eqref{eq: special case} can be considered a limit system of \eqref{eq: chemotaxis} as \( \varepsilon \to 0 \). Previous studies have explored finite-time blow-up phenomena in special cases of \eqref{eq: chemotaxis} \citep{Winkler2010}. While the solution of \eqref{eq: chemotaxis} exists globally, its height depends on \( \varepsilon \), 
so we may expect the solution or its derivative of \eqref{eq: special case} blows up when $\eps\to 0$ \citep{ZAWang2021,Choi2024}. 
\end{remark}
\begin{remark}
 To accurately observe the dispersal behavior of animals, it is preferable to exclude population dynamics and ensure that the total number of animals within the domain remains conserved. This conservation can be expressed as
\[
\int_\Omega u(x,t) \,dx = \int_\Omega u_0(x) \,dx.
\]
To maintain a constant total population over time, we mainly discussed the case with a homogeneous Neumann boundary condition:
\[
\nabla u \cdot \nu = 0 \quad \text{on } \partial\Omega \times (0,T],
\]
where \( \nu \) is the outward normal vector to the boundary \( \partial\Omega \). This boundary condition ensures that no flux of \( u \) crosses the boundary. Notice that memory quantity $g*_tu$ also satisfies the homogeneous boundary condition.

\end{remark}

\begin{remark} Note that the existence result also holds with other boundary conditions, such as the Dirichlet and the Robin boundary conditions. In the modeling perspective, the homogeneous Dirichlet boundary condition represents a hostile boundary where animals die on the boundary. In this case, the total population, however, goes to 0 as $t\to \infty$ unless it has population dynamics.
By including a suitable reaction term, such as \( f(x, u) \in C^{1,3}\left(\mathbb{R}; \bar{\Omega} \times \mathbb{R}\right) \), which satisfies the condition that there exists a constant \( \bar{u} \) such that \( f(x, u) < 0 \) for any \( x \in \Omega \) when \( u > \bar{u} \), and \( f(x, u) \geq 0 \) when \( 0 \leq u < \bar{u} \), local existence can also be established using a similar argument. Moreover, global existence is expected to hold if \( f \) is bounded.
\end{remark}

\section{Linear stability of constant steady states for a special case}
\label{app:linear stability}
We investigate the linear stability of constant steady states for system \eqref{eq: special case}, which is a special case of the model \eqref{eq: general model no f} with a weak kernel $g(t) = \tau^{-1}e^{- t/\tau}$. 
Let $(\bar u, \bar v)$ be a constant steady-state of the model \eqref{eq: special case}. Then it satisfies $\bar u = \bar v$. 
Let us define $$\chi(s): = -\frac{A( s)s}{ D( s)}.$$
\begin{lemma}\label{lem: stability}
The constant steady state $(\bar u,\bar v)$ of \eqref{eq: special case} is linearly unstable if $\chi( \bar u)>1$. It is linearly stable if $\chi(\bar u)<1$.
\end{lemma}
\begin{proof}
If we linearize \eqref{eq: special case} around $(\bar u, \bar v)$, we obtain the perturbed system
\begin{equation*}
    \partial_t\begin{pmatrix}
        u \\ v
    \end{pmatrix} = \begin{pmatrix}
        D(\bar v ) \Delta & A(\bar v)\bar u \Delta \\
        1 & -\tau^{-1}
    \end{pmatrix}\begin{pmatrix}
        u \\ v
    \end{pmatrix}=: M(\bar u,\bar v) \begin{pmatrix}
        u \\ v
    \end{pmatrix}.
\end{equation*} Let $(\mu,\mathbf{x})$ be an eigenpair of $-\Delta$ under the homogeneous Neumann boundary condition. Consider a matrix 
\begin{equation*}
    B(\bar u,\bar v):=\begin{pmatrix}
        -D(\bar v ) \mu & -A(\bar v)\bar u \mu \\
        1 & -\tau^{-1}
    \end{pmatrix}.
\end{equation*} Let $(\lambda, \mathbf{c})$ be an eigenpair of $B$. Then 
\begin{equation*}
    A(\bar u, \bar v)(\mathbf{x}e^{\lambda t}\mathbf{c}) = B(\bar u, \bar v)(\mathbf{x}e^{\lambda t}\mathbf{c}) = \lambda \mathbf{x}e^{\lambda t}\mathbf{c} = \partial_t (\mathbf{x}e^{\lambda t}\mathbf{c}),
\end{equation*} and hence $\mathbf{x}e^{\lambda t}\mathbf{c}$ is a solution of $A$. Therefore, the sign of eigenvalues of $B$ determines the linear stability of \eqref{eq: special case}.
We have $\text{tr}(B)= -D(\bar v)\mu -\tau^{-1}<0$ and $\det(B) = \mu(\tau^{-1}D(\bar v) + A(\bar v)\bar u)$. Therefore, $(\bar u,\bar v)$ is linearly unstable if and only if there exists $\mu>0$ such that $\det(B)<0$. For any positive eigenvalue $\mu$ of $-\Delta$ under the homogeneous boundary condition, if $\chi(\bar u)>1$, $\det(B)<0$ because of the relation $\bar v = \bar u$. This concludes our lemma.
    
\end{proof}

\bibliographystyle{apalike2}
\bibliography{reference}

\end{document}